\documentclass[letterpaper, 10 pt, peerreviewca]{ieeeconf}

\IEEEoverridecommandlockouts
\usepackage{graphicx,subfigure}
\usepackage{epstopdf}
\usepackage{tikz}
\usepackage{amsmath}
\usepackage{multirow}
\usepackage{amsfonts}
\usepackage{float}
\usepackage{amssymb}
\usepackage{graphicx}
\newtheorem{theorem}{Theorem}
\newtheorem{lemma}{Lemma}
\newtheorem{proposition}{Proposition}
\newtheorem{corollary}{Corollary}
\newtheorem{definition}{Definition}
\newtheorem{remark}{Remark}
\makeatletter

\newcommand{\Rmnum}[1]{\expandafter\@slowromancap\romannumeral #1@}
\makeatother

\begin{document}
\title{\LARGE Power Grid Decomposition  Based on Vertex Cut Sets and Its Applications to Topology Control and Power Trading}
\author{Shuai Wang \& John Baillieul}
\maketitle
\let\thefootnote\relax\footnotetext{\noindent\underbar{\hspace{0.8in}}\\
Shuai Wang is  with the Division of Systems Engineering and John Baillieul is with the Departments of Mechanical Engineering; Electrical and Computer Engineering, and the Division of Systems Engineering at Boston University, Boston, MA 02115. Corresponding author is John Baillieul (Email: johnb@bu.edu). \newline The authors gratefully acknowledge support of NSF EFRI grant number 1038230. }
\begin{abstract}
It is well known that the reserves/redundancies
built into the transmission grid in order to address a variety
of contingencies over a long planning horizon may, in the
short run, cause economic dispatch inefficiency. Accordingly, power grid optimization by
means of short term line switching has been proposed and is
typically formulated as a mixed integer programming problem
by treating the state of the transmission lines as a binary decision variable, i.e. in-service
or out-of-service, in the optimal
power flow problem. To handle the combinatorial explosion, a
number of heuristic approaches to grid topology reconfiguration
have been proposed in the literature. This paper extends our
recent results on the iterative heuristics and proposes a fast
grid decomposition algorithm based on vertex cut sets with
the purpose of further reducing the computational cost. The paper concludes with a discussion of the possible relationship between vertex cut sets in transmission networks and  power trading.
\end{abstract}
 \section{Introduction}
Siting and maintaining massive power infrastructure is not cheap \cite{vajjhala} and therefore makes the optimal use of existing network a priority.
Because of the fast time constants in changing the system state and the very low costs, corrective  switching operations, including transmission line switching, bus-bar switching, and shunt element
switching, etc.,  are often the first post-contingency corrective action to be considered and implemented \cite{glavitsch}.

 The  focus of corrective  switching has been mainly on  handling line overload \cite{Mazi,Granelli} and voltage violations \cite{Bakirtzis,Rolim}, and most recently, on  co-optimizing the generation along with the network topology \cite{Fisher,Ruiz,Fuller}.
 Based on the DC optimal power flow (OPF) problem, such co-optimization is typically formulated as a mixed integer programming problem (MIP) with some binary variables denoting the choices of transmission line switches. There are a number of challenges in implementing effective transmission topology control, including understanding and avoiding voltage problems, transient instability, reactive power problems \cite{Hedman}  and, most importantly, the ability to find a good solution within a time that is short enough to be of practical use. 
 
Due to the combinatorial nature of the problem and the nonlinearities inherent to networks, the approaches proposed in the literature
focus on the reduction of the search space of possible switching actions \cite{Hedman2}. 
Accordingly, a set of fast line switching heuristics have been proposed, and the  running time records have been consistently set and then shattered  \cite{Fisher,Ruiz,Fuller}. In all of these studies, simulations show that significant savings are indeed possible from transmission switching. While optimality for this NP-hard problems
is ever elusive,  theorems together with simulation results  on a set of iterative
heuristics in \cite{preprint}, nevertheless, offer some justification that, with similar searching space, the locally optimal switching at each stage of an iterative heuristic has the best chance in reducing generation cost while enforcing connectivity.

The present paper examines possible grid decomposition approaches that are able to further reduce the search space of possible switching operations while at the same time retain the locally optimal switchings  to the maximum degree.   The paper is organized as follows. 
  In Section II, we first discuss the typical limitations with respect to the {\em composability} and {\em composationality} of
real transmission grids, showing the challenges in grid decomposition. 
A series  of trade-off is then proposed, and this forms the main thread of the paper. Section III proposes a decomposition method based on  {\em vertex cut sets} and {\em pseudo biconnected components}.  Rather than intuitively using {\em line outage distribution factor} to quantify the decomposition quality, Section IV makes a further trade-off by justifying the usage of {\em locational marginal price range}   to quantify the  potential of a pseudo biconnected component in reducing generation cost. 
  The useful properties of the vertex cut set in LMP partitioning and power/FTR trading are   discussed in great detail  in Section V. A fast grid decomposition algorithm  is then proposed in Section VI, and its efficiency is tested in Section VII. The paper ends with concluding remarks   in Section VIII.
 
\section{Composability and Compositionality}
When proposing heuristics to deal with the ``largest machine in the world", the transmission network, one natural idea is a divide-and-conquer strategy. We ask whether this large and complex
system can be usefully decomposed into
smaller, more tractable subsystems, thereby reducing the  search space of line  switching.
Here, we refer the ease of building a system  out of subsystems as $composability$, and the ease of validation of system properties by using the related subsystems  as $composationality$ \cite{Kammerer}. 

A transmission network can be modeled as a graph $G=(V,E)$ such each node in $V$  is either a substation bus, or simply a  bus connecting several  transmission lines. There is an edge $e_{i,j} = (v_i , v_j )\in E$ between two nodes if there is a physical line connecting directly the elements represented by $v_i$ and $v_j$. Note that the terms vertex/node/bus, link/line/edge and grid/graph/network will be used interchangeably.

In graph theory, a {\em biconnected component} in a graph is one where if any node is removed (with its boundary edges), the component will still be connected \cite{Chartrand}.  Reference \cite{Wang} proves that  the impact of topology control is strictly confined within the biconnected component  in which the line switching is applied in a transmission network. The decomposition of a transmission network based on biconnected components is thus ``clean" in the sense that  if the sizes of all its biconnected components are relatively small,  the network will have high level composability and composationality.  

Also proved in \cite{Wang}, the event occurs  with probability 1 that disconnecting a transmission line will change every line flow within the biconnected component that contained the disconnected line. The  N-1  reliability standard of today's transmission network, requiring the system to move  to a satisfactory state for any single line outage event,    significantly increases the likelihood of the whole network or a large portion of the network being biconnected.  Despite of the desirability of maintaining N-1 connectivity, the downside of   large biconnected components is a significantly decreased level of the network composability and composationality. 

While we should not neglect that every switchable line has the potential to alleviate the congestion in a biconnected component of a grid, simulation results in \cite{Ruiz} showed  that a fairly large portion of the switchable lines remained not operated in most samples, and thereby might well go unnoticed. Such observations suggest  the possibility that a biconnected component of a grid may  still be effectively  further decomposed  even if the interactions between  the sub-grids of the biconnected component are  not  zero.

Our work  thus aims to find an optimal trade off between the size of sub-grids of the biconnected component and the realized composability/composationality  by the grid decomposition.  The desired outcome should help us focus  attention on a relatively small set of promising switchable lines such that both the computational efforts and  the attainable savings can be reasonably controlled.

\section{Vertex Cut Sets \& Pseudo Biconnected Components}\label{vcspbc}
The  initially ``clean" decomposition of a transmission grid  involves a  graph theory term called the \emph{cut vertex} \cite{Chartrand} which is any single vertex/node whose removal increases the number of connected components of a graph and, by definition, is the only element shared by two or more biconnected components of  a graph.   

Apparently, further decomposing a biconnected component needs something beyond the cut vertex, and we propose the following:
\begin{definition} 
 A  \emph{vertex cut set} is a set of vertices  of a graph which, if removed together with any incident edges, disconnects the graph.  The sub-graphs of a biconnected graph obtained by one or several  vertex cut sets are called  \emph{pseudo biconnected components}. The vertices  in the vertex cut sets are the only elements shared by two or more pseudo biconnected components. In addition, the union of  several vertex cut sets for a given graph is also a vertex cut set.
\end{definition}

Using the concept of vertex cut sets, a two-step construction of pseudo biconnected components can be carried out. Given a graph and a vertex cut set, the disconnected  subgraphs remaining when the vertex cut set,  together with  any incident edges,  is removed are called the pre pseudo biconnected components. From these, the pseudo biconnected components are obtained by adjoining the vertex cut set to each pre pseudo biconnected component the entire vertex cut set along with the incident edges that have been removed. Fig. \ref{decompose} shows an example of the decomposition of a biconnected graph where the three red nodes in the middle  of graph form the vertex cut set. Note that the decomposition based on one vertex cut set may not be unique. For example, the horizontal black edge in Fig. \ref{decompose}(a) can be either assigned to the black-edge sub-graph (Fig. \ref{decompose}(c1)) or  to the blue-edge sub-graph (Fig. \ref{decompose}(c2)).  Thus we do have  some degree of discretion  in decomposing the grid. To reduce the search space of line switching, we will use the convention throughout this paper that, whenever possible, any incident edges of the vertex cut set subject to our discretion  will be assigned to a pseudo biconnected component that is free of congested lines.
\vspace{-0.15in}
\begin{figure}[htbp]
	\centering
		\includegraphics[width=0.5\textwidth]{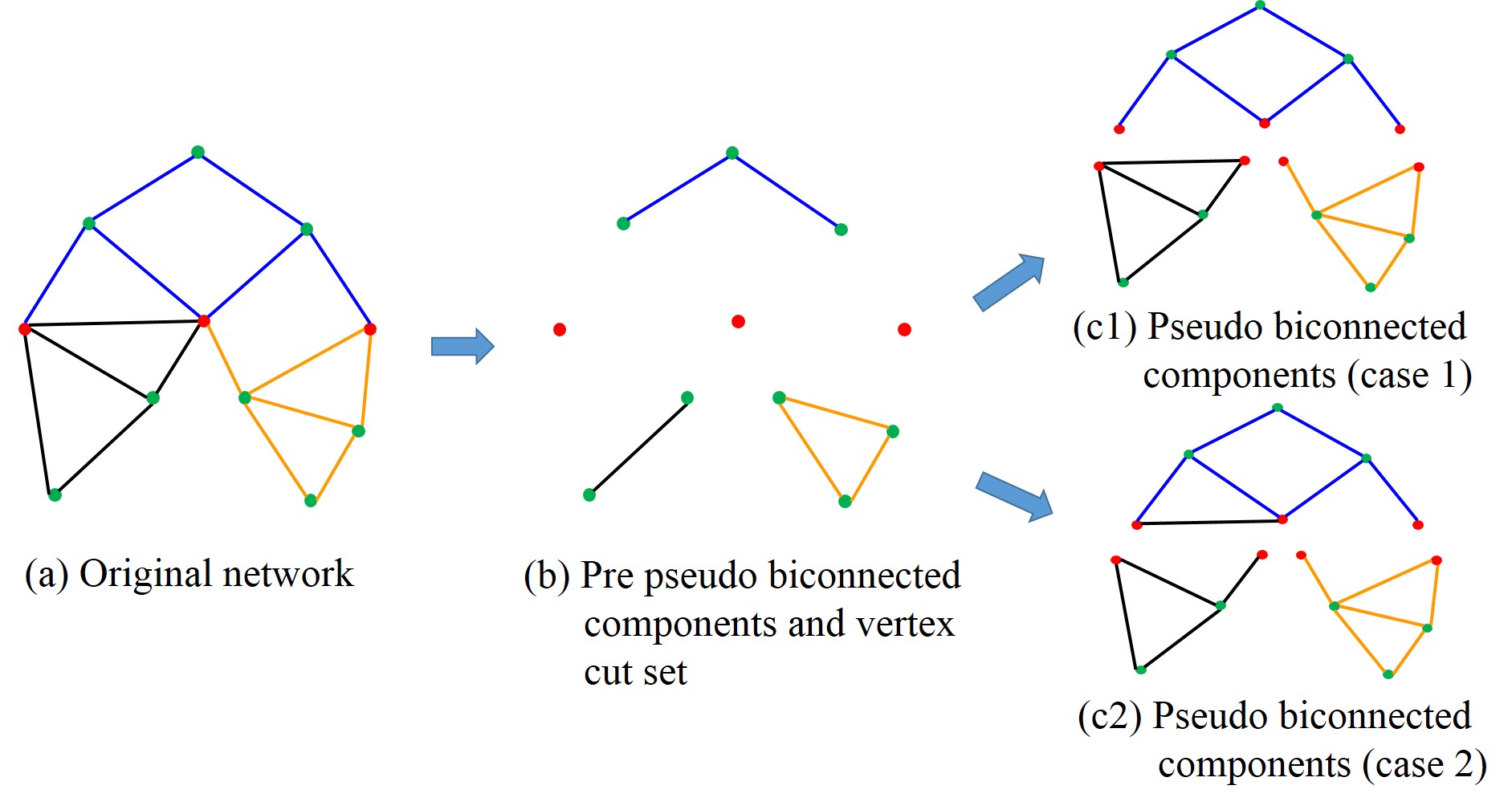}
	\caption{One example of the decomposition based on a vertex cut set formed by the three red nodes in the middle.} \label{decompose}
\end{figure}

Ideally, we would expect the decomposition in Fig \ref{decompose} to satisfy the following criteria:
\begin{itemize}
\item[(1)]  The interactions between any two lines belonging to two different pseudo biconnected components are  insignificant;
\item[(2)]  The size of the targeting pseudo biconnected component  is large enough to contain all promising switchings;
\item[(3)]  The size of  the targeting pseudo biconnected components is small enough to be   efficiently searched.
\end{itemize}

\section{Quantifying Decomposition Quality}
So far, the interactions mentioned in Criterion (1) at the end of Section \ref{vcspbc} are imprecisely specified, and hence we look for some metrics to quantify the effects of   line outages.

We introduce the well-known sensitivity, \emph{line outage distribution factor} or  LODF, denoted by $\zeta _i^{(j)}$,  of the line flow $i$  with respect to the removal of line $j$.  In short,  $\zeta _i^{(j)}$ denotes the percentage of  line flow on line  $j$ that will show up on line $i$ after the outage of line $j$.   Similarly, we can compute the LODF, denoted by $\zeta _j^{(i)}$,  of the power flowing through line $j$ with respect  to the removal of line $i$. The pair of LODFs, $\zeta _i^{(j)}$ and  $\zeta _j^{(i)}$,  are usually different. The interactions between the line pair $\{i,j\}$ are said to be  insignificant if   $\zeta _i^{(j)}$ and  $\zeta _j^{(i)}$ are both small. Similarly, the interactions between two pseudo biconnected components are said to be  insignificant if the pair of LODFs for any two lines belonging to  two different pseudo biconnected components are small.

To quantify the quality of  grid decomposition based on vertex cut set, we need to recalculate the  LODFs for each  possible line switching which could be  computationally intensive for large scale grids.

The computational overhead  pushes us to make another trade off between the  computational costs and the preciseness of the metrics measuring the interactions in a grid. Such trade off is achieved by introducing    another well-know sensitivity, \emph{power transfer distribution factor} or \emph{PTDF}, which gives the sensitivity of the flow on one line  with respect to a power transfer between a pair of buses.

The PTDF, by definition, is essentially the difference between a pair of \emph{shift factors}. 
In the DC power flow model, the \emph{transmission sensitivity matrix} $\Psi$, also known as the \emph{injection shift factor matrix}, gives the sensitivities in line flows due to changes in the nodal injections, with the reference bus assumed to ensure the power balance. Suppose $\Psi_{i,j_1}$ denotes the sensitivity in power flowing through  line $i$ due to one unit increase in the power injection at  bus $j_1$ and one unit decrease in power injection at the reference bus, i.e.
\[\Psi_{i,j_1}= \frac{\Delta flow \ in \ line \ i }{\Delta injection \ at \ bus \  j_1} ,\] then   the  shift factor difference $ |\Psi_{i,j_1}-\Psi_{i,j_2}| $ gives the PTDF of power flowing through line $i$ with respect to  one unit power transfer between  bus $j_1$ and    $j_2$. 

The network LODFs, mathematically, are equivalent to the post-outage network PTDFs. While the pre- and post-outage PTDFs are rarely equivalent, there usually exists a strong correlation between their values. Therefore, although the estimations of the decomposition quality based on pre-outage PTDFs maybe  different from those based on LODFs, such differences are reasonably expected to be small.

Here, we would like to remark that the introduction of PTDF  does   reduce the computation cost associated with the estimations of decomposition quality, but such computation effort is still relatively small compared with solving the updated DCOPFs. The value of the following PTDF analysis, in fact, is more towards a fundamental understanding of some structural properties of the grid decomposition, thus reducing the computational effort spent on finding the optimal vertex cut sets.

We now need to deal with the most challenging step  of the decomposition, i.e.  minimizing the interactions (PTDFs) between pseudo biconnected components. Though highly desired, such decomposition quality is usually not achievable by just using vertex cut sets as the cushion zones between pseudo biconnected components, due to the  unavoidable strong LODFs between those edges belonging to different pseudo biconnected components but incident to  same nodes in the vertex cut sets. Instead of  complicating the decomposition by expanding the cushion zones   to some sub-graphs containing both vertex cut sets and the associated incident edges, we  would like to stick to the vertex cut sets, indicating a necessity to modify the criteria.  Rather than struggling with the nominal PTDFs between pseudo biconnected components, our decomposition instead aims to minimize the  potentially ``useful" portion of those PTDFs, i.e. the portion associated with the congestion effect.
 
When talking about the congestion effect,  it is natural to wonder about the 
the $locational$ $marginal$ $price$ (LMP). In a lossless DC power flow model, the LMP at a given bus, say bus $i$, is equal to the sum of the LMP at the reference bus and the congestion cost which can be calculated using the following equation:
\begin{equation} \label{lmp_formula}
\lambda_i=\lambda_{ref}-\sum_j(\mu_j\times \Psi_{j,i})
\end{equation}
where 
\[
\lambda_i=LMP\ at\ bus\ i
\]
\vspace{-0.27in}
\[
\lambda_{ref}=LMP\ at\ the\ reference\ bus
\]\vspace{-0.27in}
\[
\mu_j=Shadow\ price\ of\ line\ constraint\ j
\]\vspace{-0.27in}
\[
\Psi_{j,i}= Shift\ factor\ for\ real\ load\ at\ bus \ i\ on\ line\ flow\ j.
\]

The sensitivity analysis of DCOPF model suggests that power transfer between  buses with close LMPs  usually can't change the generation cost as much as that between buses with significantly different LMPs, and similarly line switches in areas with low range LMPs are typically of little use in reducing overall generation costs. Therefore,  the goal of the partition problem based on vertex cuts  is  to decompose the grid into two   sub-grids with one    being of high range LMPs and the other being of  low range LMPs.

\section {The Optimality of Vertex Cut Sets} \label{Sec:FTR}

In this section, we will show a set of useful properties of vertex cut sets with respect to the decomposition based on LMPs. The discussion starts with the following definition.
\begin{definition}Suppose there is one line congestion  in a transmission network  and we are able to control certain number of nodal injections, then the smallest set (i.e. of smallest cardinality) of buses  that  is able to fully alleviate the given  line congestion by changing least amount of nodal power injections is defined  as the $optimal$ set of buses for that congested line.
\end{definition}

\begin{remark} The optimal  set of buses for a given congested line may not be unique. 
\end{remark}

We first show that for a given congestion in one pseudo biconnected component, the optimal set of buses in another pseudo biconnected component must be a subset of the vertex cut set.   The complete proof involves two lemmas, one theorem, and one proposition.

We start the proof by introducing the well established  correspondence between power grids  and  electric   circuits. 
 Mathematically, the DC   power flow model is  equivalent to a current driven network (an electric   circuit comprised purely by resistors and current sources) \cite{Wang}, where power injections are equivalent to  current sources; power flowing through lines is equivalent to current through edges, etc.  See Table 1. Thus part of the following discussion will be focusing on   resistance networks with nodal current injections controlled by a set of current sources. 

\begin{center}
\begin{tabular}{l@{\hskip 0.1in}c@{\hskip 0.1in}c@{\hskip 0.1in}c@{\hskip 0.1in}c}\hline
 network &  potential   &  injection  &  admittance  & equation \\ \hline 
 circuit & voltage $V$ & current $I$ & conductance $G$ & $I=GV$   \\  \hline
grid&  phase $\theta$  & power $P$ & susceptance $B$ &   $P=B\theta$ \\  \hline
\end{tabular}\end{center}\begin{center}
{Table 1: The equivalence between  current driven circuits and DC power flow models of transmission grids.}
\end{center} 

\begin{lemma} If we are allowed to alter all controllable nodal power injections by arbitrary amount and all uncongested lines have enough line capacity, then the cardinality of the optimal  set of buses for a given  congested line must be 2, i.e. the  optimal  set of buses is a bus pair.
\end{lemma}
\begin{proof} Shift factors do not change when the injection/withdrawal amount increases for any set of buses. Suppose line $i$ is overloaded by $\Delta_i$ units of power and the power injections at bus $j_1$ and $j_2$ are controllable. In order to alleviate the congestion of line $i$, the power injections at bus $j_1$ and $j_2$ need to be changed by
\[\Delta_{j_1\& j_2}=\frac{\Delta_i}{|\Psi_{i,j_1}-\Psi_{i,j_2}|}.\] Thus the pair of controllable buses with the largest shift factor difference associated with the congested line  must be the optimal set of buses for that congested line.
\end{proof}

\begin{theorem}
(\cite{Dorfler}) Any n-terminal connected resistance network is electrically equivalent to a  circuit  whose graph is a complete graph consisting of n(n-1)/2 resistances. Such n-terminal equivalent graph  can be obtained by using Kron Reduction to  eliminate all interior (non-terminal) nodes.
\end{theorem} 

\begin{remark} The above theorem can be easily extended to any  disconnected graph by adding some trivial lines of infinite resistance to make the equivalent graph connected and complete. The current flowing through those trivial lines must always be zero no matter how   the nodal current injections are changed.
\end{remark}

\begin{lemma} Suppose line $i$ is connected to bus pair $\{i_1,i_2\}$ and line $j$ is connected to bus pair $\{j_1,j_2\}$   in a connected transmission network, and  the reference directions of the power flowing through line $i$ and line $j$  are defined  as from bus $i_2$ to bus $i_1$ and  from bus $j_2$ to bus $j_1$, respectively.   Then the shift factor differences $\Psi_{i,j_2}-\Psi_{i,j_1}$ and $\Psi_{j,i_2}-\Psi_{j,i_1}$ must be of the same sign.
\label{lemmashift}
\end{lemma} 
\begin{proof}By keeping applying Kron reduction \cite{Dorfler} to the original transmission network (shown in the right side of Fig. \ref{demo2}) until all nodes but  $\{i_1,i_2,j_1,j_2\}$  are eliminated, we can create a 4-terminal equivalent grid with only four buses $\{i_1,i_2,j_1,j_2\}$ which is shown in the left side of Fig. \ref{demo2}. Note that the two red lines in Fig. \ref{demo2} may be of different susceptance  but the direction  of power flowing the lines, $P_j$ and $P_{j^{'}}$, must be the same. The same is true of the two blue lines. Thus we   focus our attention on the 4-terminal equivalent grid.

\begin{figure}[htbp]
	\centering
		\includegraphics[width=0.4\textwidth]{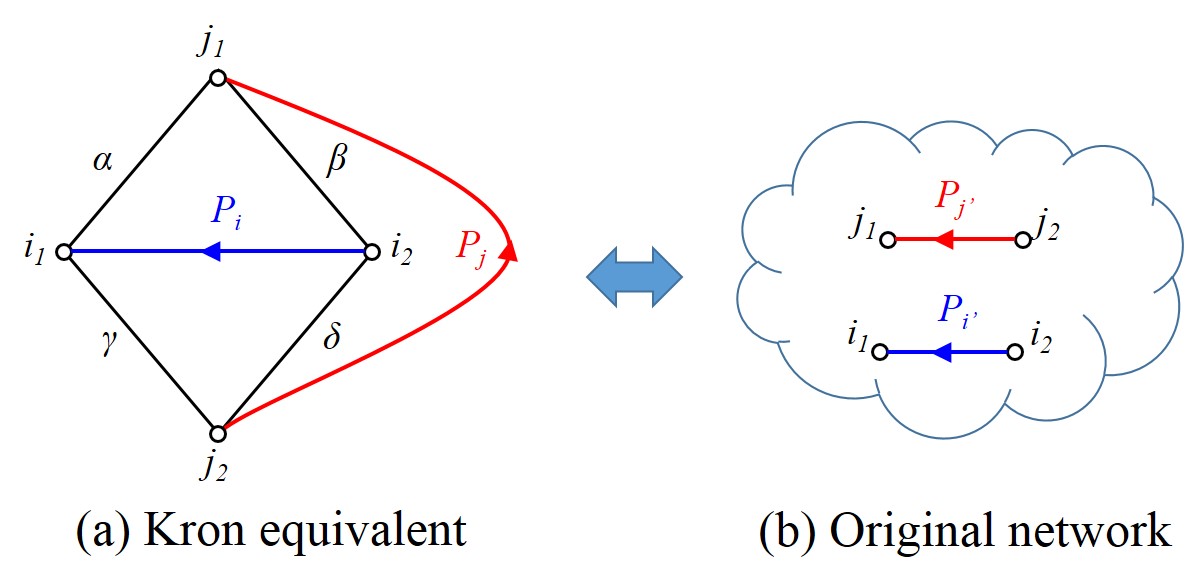}
	\caption{The 4-terminal equivalent network of the original network.} \label{demo2}
\end{figure}

By the equivalence in Table 1, we  can actually  treat the 4-terminal equivalent grid as a 4-terminal  resistance network and discuss the problem in a circuit environment. Under the circuit environment, we denote the  line resistances of the four black lines in the left side of Fig. \ref{demo2} as $\alpha$, $\beta$, $\gamma$ and $\delta$, respectively, and suppose there is no current injections and withdrawals in the network except   one unit  nodal injection at node $i_2$ and one unit   nodal withdrawal at node $i_1$. If the  direction of red line flow in the left figure  is from node $j_2$ to node $j_1$, we know $\Psi_{j,i_2}-\Psi_{j,i_1}>0$, indicating 
\[\frac{\delta}{\gamma}<\frac{\beta}{\alpha},\] that we must have $\Psi_{i,j_2}-\Psi_{i,j_1}>0$. If the  direction of red line flow in the left figure  is from node $j_1$ to node $j_2$, we know  $\Psi_{j,i_2}-\Psi_{j,i_1}<0$, indicating \[\frac{\delta}{\gamma}>\frac{\beta}{\alpha},\] that we must have $\Psi_{i,j_2}-\Psi_{i,j_1}<0$.
\end{proof}

\begin{proposition}\label{SFVCS}
Suppose we use a  vertex cut set to decompose a connected transmission network into $n$ ($n\geq 2$)  sub-graphs. Assume that we can arbitrarily change the power injection at all buses in the $j$-th sub-graph to alleviate a given congested  line  in the $i$-th sub-graph ($i\neq j$),  then one optimal pair of buses must be a subset of the vertex cut set.
\end{proposition}
\begin{proof}
Although the value of each entry in the injection shift  factor matrix varies as   the reference bus changes,  the PTDF or the difference between a pair of shift factors  is independent of  the  choice of the reference bus.  Suppose line $i$ is congested and is connected to bus $i_1$ and $i_2$, and we pick up $i_1$ as the reference bus and assuming the direction  from $i_2$ to $i_1$ as the positive power direction, we can get an injection shift factor matrix with all $i$-th row entries being non-negative. 

Suppose there are $M$ buses $\{n_1,...,n_M\}$ that are shared by   the $j$-th sub-graph and the rest of the graph. Clearly, $\{n_1,...,n_M\}$  is a subset of the vertex cut set.  Then in order to prove the proposition, we just need to show that  the shift factor differences $\{\Psi_{i,p}-\Psi_{i,n_1},..., \Psi_{i,p}-\Psi_{i,n_M}\}$ cannot be all positive or all negative for any bus $p$ that is owned exclusively by the $j$-th sub-graph, i.e. 
\[
min\{\Psi_{i,n_1}, ...,\Psi_{i,n_M} \}\leq\Psi_{i,p}\leq max\{\Psi_{i,n_1}, ...,\Psi_{i,n_M} \}.\] 

We first  create two sub-networks to simplify the proof.  The first sub-network  (shown at the bottom middle of Fig. \ref{demo}) is created from the original network by removing all lines in the $j$-th sub-graph, and it is called the nearby sub-network. The second sub-network (shown at the top middle of Fig. \ref{demo}) is just the $j$-th sub-graph,  and it is called the remote sub-network. It is easy to see that the two sub-networks share nothing except  the buses  $\{n_1,...,n_M\}$. Buses $\{i_1,i_2\}$ are interior buses of the nearby sub-network, and bus $p$ is interior bus of the remote sub-network.  

Then an equivalent sub-network of the remote sub-network can be created by applying Kron reduction to the remote sub-network until all nodes but  $\{n_1,...,n_M,p\}$ are removed. The equivalent sub-network obtained in this way is a $K_{M+1}$ complete graph as shown at the top right corner of Fig. \ref{demo}. 
An equivalent network of the original network can then be created by putting together the nearby sub-network and the $K_{M+1}$ equivalent sub-network. The process described above is visualized in Fig. \ref{demo}. 
\begin{figure}[htbp]
	\centering
		\includegraphics[width=0.6\textwidth]{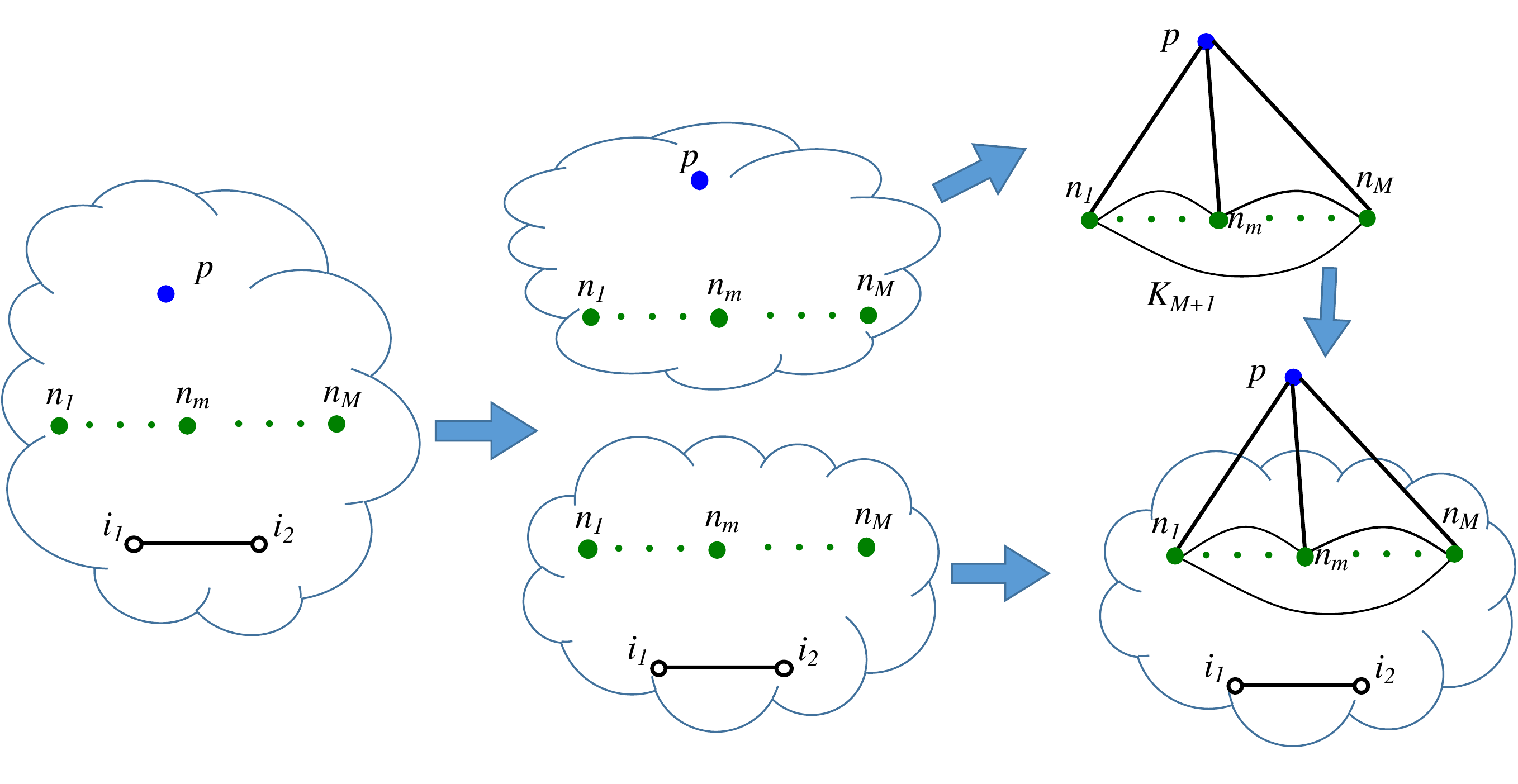}
	\caption{The network simplification process.} \label{demo}
\end{figure}

Suppose  the reference direction of the power flowing through the line connecting bus $p$ and  $n_m$ ($m=1,...,M$) in the equivalent network (shown at the bottom right corner of Fig. \ref{demo}) is defined as from bus $p$ to bus $n_m$. Then by Lemma \ref{lemmashift}, we know that the change of power flowing through the line connecting bus $p$ and  $n_m$ ($m=1,...,M$) in the equivalent network due to one unit increase in nodal injections at bus $i_2$ and one unit decrease in nodal injections at bus $i_1$ is of the same sign as that of the shift factor difference $\Psi_{i,p}-\Psi_{i,n_m}$.

Finally, we denote the power flowing through the lines connecting node $p$ and node $n_m$ ($m=1,...,M$) as $\{I_{p,n_1},...,I_{p,n_M}\}$  due to one unit  nodal injection  at bus $i_2$ and one unit   nodal withdrawal at bus $i_1$. By Kirchhoff's current law, we must have \[\sum_{m=1}^{M}I_{p,n_m}=0\]
 meaning \[\sum_{m=1}^{M}c_m(\Psi_{i,p}-\Psi_{i,n_m})=0\]
 where $I_{p,n_m}=c_m \Psi_{i,n_m}$ and $c_m$ are some positive coefficient by Lemma \ref{lemmashift}. Thus $\{\Psi_{i,p}-\Psi_{i,n_1},..., \Psi_{i,p}-\Psi_{i,n_M}\}$ cannot be all positive or all negative. 
\end{proof}

Rather than limiting the discussion to single line congestion, we'd like to know if phenomena similar to Proposition \ref{SFVCS} can be found for scenarios with multiple congestions. Things become a little  complicated here as the changes of line flows of the congested lines due to the nodal injection regulations may not  be all favorable, i.e. the alleviation of  some congestions may come inevitably at the expense of worsening other congestions, requiring all congestions being assigned proper weights to evaluate the overall benefit of the nodal injection regulations. Since our main consideration is to reduce generation cost,  the shadow prices of  line constraints are naturally  used as the weights in the following discussion.

We will show that for multiple given congestions in one pseudo biconnected component, the optimal set of buses in another pseudo biconnected component must also be a subset of the vertex cut sets.   The complete proof involves  one proposition and two corollaries.
\begin{proposition}\label{vcs}
Suppose we use a  vertex cut set to decompose a connected transmission network into $n$ ($n\geq 2$)  sub-graphs. Assume that we can arbitrarily change the power injections at   buses in  the $j$-th sub-graph to alter the line flows  in the $i$-th sub-graph ($i\neq j$),  then such impact on the  $i$-th sub-graph can  be precisely replicated/canceled by  regulating the power injections of the buses in the vertex cut set.
\end{proposition}
\begin{proof}
By  the superposition principle in the DC power flow model, we just need to  prove that  all line flows in the $i$-th sub-graph can always be set to zero  if the net power injections at all interior buses of the $i$-th sub-graph are strictly zero and if we can arbitrarily regulate the power injections at the buses in the vertex cut set.  

We first show that the reduced nodal susceptance matrix $B$ of any transmission network is positive definite. Since any resistance network is a passive system that consumes energy, and the absorbed power is given by 
\begin{equation}
P=v^{T}Gv
\end{equation}
where $v$ is the nodal voltage vector and $G$ is the reduced conductance matrix of the resistance network. It follows that the energy consumption will always be positive only if all eigenvalues of $G$ are positive. Thus, a criterion for passivity is that $G$ be positive definite. Since  the DC model of power flow is equivalent to a current driven network as shown in Table 1, we know the reduced  susceptance matrix $B$ in the DC power flow model is  positive definite, meaning its inverse matrix $B^{-1}$ is also positive definite. 

Suppose, in the bus vector, the first $l$ buses are owned exclusively by the $i$-th sub-graph, the middle $m$ buses are owned by the vertex cut set, and the last $n+1$ buses containing the reference bus are owned exclusively by other sub-graphs,  respectively, we then have the DC power flow equation
\begin{equation}
\theta=B^{-1}p
\end{equation}
where $p$ is the net nodal power injection vector, and $\theta$ is the nodal voltage angel vector. Based on the classification of the buses, the reduced voltage angle vector $\theta$ can be decomposed into 3 sub-vectors
\[\theta^T=[\theta_l^T\ \  \theta_m^T\ \ \theta_n^T],\]
and the positive definite matrix $B^{-1}$ can be decomposed into 9 sub-matrices 
\begin{equation}
\begin{split}
&\begin{matrix}
\quad\ \ \ \ \ \  \ \ \ in &\quad\ at &\quad\  \ out
\end{matrix}\\&
\begin{matrix}
\quad\quad\quad\ \ \ \ \downarrow &\quad\ \ \downarrow &\quad\ \ \ \ \downarrow
\end{matrix}\\&
B^{-1}=\begin{bmatrix}
A_{l\times l} &  B_{l\times m} & C_{l\times n} \\ 
D_{m\times l}  & E_{m\times m}  & F_{m\times n} \\ 
G_{n\times l}  & H_{n\times m}  & I_{n\times n} 
\end{bmatrix}
\begin{matrix}
\leftarrow \\ 
\leftarrow \\ 
\leftarrow 
\end{matrix}
\begin{matrix}
in \\ 
at \\ 
out 
\end{matrix}
\end{split}
\end{equation}
and we know the sub-matrix $E_{m\times m}$ must have full rank and be positive definite.  Thus we have
\[\theta_m=\begin{bmatrix}
D_{m\times l}  & E_{m\times m}  & F_{m\times n}
\end{bmatrix}p.\]
 
It follows that we can always set the  voltage angles    of the vertex cut set $\theta_m$ to be equal  if we can arbitrarily regulate the power injections of the buses in the vertex cut set. 

 Next we are going to prove that the   line flows in the $i$-th sub-graph must be all zero if we set  the  voltage angles  $\theta_m$   of the vertex cut set to be equal and  if the net power injections at the interior buses of the $i$-th sub-graph are all zero.
If under such scenario there still exist some non-zero line flows in the $i$-th sub-graph,  then we  know that the voltage angles of the interior buses of the $i$-th sub-graph can not be equal, meaning at least one of the interior buses is of non-zero net power injection. This  contradicts our assumption. 
\end{proof}
\begin{corollary}
Suppose we use a  vertex cut set to decompose a connected transmission network into $n$ ($n\geq 2$) connected sub-graphs. We first classify the columns of  the injection shift factor matrix into 3 sets with the first set being associated with the buses owned exclusively by the $i$-th sub-graph, and the second set being associated with the buses in the vertex cut set, and the third set  being associated with the buses owned exclusively by other sub-graphs. Assume their cardinalities are $l$, $m$, and $n$, respectively. We then  classify the rows of  the injection shift factor matrix into 2 sets with the first set being associated with the lines owned exclusively by the $i$-th sub-graph, and the second set being associated with the lines owned   exclusively by other sub-graphs, and assume their cardinalities are $x$, and $y$, respectively.  A visualization of the decomposition is shown below.
\begin{equation}
\begin{split}
&\begin{matrix}
\quad\ \ \ \ \ \  \ \ \ in &\quad\ at &\quad\  \ out
\end{matrix}\\&
\begin{matrix}
\quad\quad\quad\ \ \ \ \downarrow &\quad\ \ \downarrow &\quad\ \ \ \ \downarrow
\end{matrix}\\&
\Psi=\begin{bmatrix}
A_{x\times l} &  B_{x\times m} & C_{x\times n} \\ 
D_{y\times l}  & E_{y\times m}  & F_{y\times n} 
\end{bmatrix}
\begin{matrix}
\leftarrow \\ 
\leftarrow 
\end{matrix}
\begin{matrix}
in \\ 
out 
\end{matrix}
\end{split}
\end{equation}
Then the ranks of  the sub-matrix 
\[\begin{bmatrix}
B_{x\times m} & C_{x\times n}
\end{bmatrix}\]
and 
\[\begin{bmatrix}
D_{y\times l}  & E_{y\times m}
\end{bmatrix}\]
must both be $\leq$ $m$.\end{corollary}
\begin{proof}
By Proposition \ref{vcs}, we know that the  effect of one unit net power injection  at a bus in the $j$-th sub-graph on lines in the $i$-th sub-graph is a linear combination of the effects of one unit net power injections  at the buses in the vertex cut set on lines in the $i$-th sub-graph. Thus we know each column of the sub-matrix $[C_{x\times n}]$ is a linear combination of the columns of the sub-matrix $[B_{x\times m}]$, meaning the rank  the sub-matrix 
$[
B_{x\times m} \ C_{x\times n}
]$ is $\leq$ m. Similarly, we can prove that each column of the sub-matrix $[D_{y\times l}]$ is a linear combination of the columns of the sub-matrix $[E_{y\times m}]$,  meaning the rank  the sub-matrix 
$[
D_{y\times l}  \ E_{y\times m}
]$ is $\leq$ m .
\end{proof}
\begin{corollary}\label{corr_lmp}
Suppose we use a  vertex cut set to decompose a connected transmission network into $n$ ($n\geq 2$)  sub-graphs and all congested lines are contained in the $i$-th sub-graph (which may be unconnected),  then one pair of buses with maximum LMP difference in the $j$-th sub-graph  ($i\neq j$) must be in the vertex cut set.
\end{corollary}
\begin{proof}
Here, we just need to show that the LMP at any bus in the $j$-th sub-graph is a convex combination of the LMPs at the buses in the vertex cut set.

Suppose there  are $m$ buses in the vertex cut set and their LMPs are $\{\lambda_1,...,\lambda_m\}$, and we randomly pick a bus, say bus $j1$ with LMP being $\lambda_{j1}$, from the $j$-th  sub-graph.
By Proposition \ref{vcs}, we know that the  effect of one unit net power injection  at a bus  $j1$ on lines in the $i$-th sub-graph is a linear combination of  the effect of   one unit net power injections  at the $m$ buses in the vertex cut set on lines in the $i$-th sub-graph, and we assume the coefficients of the linear combination are $\{a_1,...,a_m\}$, where the power balance equation ensures that  $a_1 +a_2 +...+a_m=1$.  Thus we have
\[ \vec{ \Psi}_{i,j1}   =[\vec{ \Psi}_{i,v_1},...,\vec{ \Psi}_{i,v_m}][a_1,...,a_m]^T,\]
where $\vec{ \Psi}_{i,j1}$ denotes the shift factor vector for read power at bus $j1$ on lines in the $i$-th sub-graph,  $\vec{ \Psi}_{i,v_k}$ ($k=1,...,m$) denotes the shift factor vector for read power at the $k$-th vertex cut set bus on lines in the $i$-th sub-graph.

 Since  we assume  all congestions are in the $i$-th sub-graph, we know the shadow prices for   line constraints that are not in the $i$-th sub-graph are all zero. Denoting the shadow price vector for all    line constraints   in the $i$-th sub-graph by $[\mu_i]$, we have
\begin{equation}
\begin{split}
\lambda_{j1} & =\lambda_{ref}-[\mu_i]^T \vec{ \Psi}_{i,j1} \\
&=\lambda_{ref}-[\mu_i]^T [\vec{ \Psi}_{i,v_1},...,\vec{ \Psi}_{i,v_m}][a_1,...,a_m]^T\\
&=a_1 \lambda_1+a_2\lambda_2+...+a_m\lambda_m
\end{split}
\end{equation}

Next we will show that  $a_i\geq0$ ($i=1,...,m$), i.e. to  cancel the effect of one unit net power injection  at  bus  $j1$ on lines in the $i$-th sub-graph  requires a non-negative  net power withdrawn at each bus in the vertex cut set.
Suppose we only inject and withdraw power at bus $j1$ (which has one unit net power injection) and the buses  in the vertex cut set, and their effects on the $i$-th sub-graph   precisely offset each other. Then the   bus with the lowest voltage angle must be one of the buses with positive net power withdrawn (or negative net power injection) and   its voltage angle must be strictly lower than any  bus with  positive net power injection. Apparently,  the bus with lowest voltage angle must be in the vertex cut set. If one of the buses in the vertex cut set has positive  net power injection, meaning its voltage angle is different from the lowest nodal voltage angle, then there exists a nodal voltage angle difference in the vertex cut set. In such case, the absorbed power in the $i$-th sub-graph \[P_i=[\theta_i^{T} \theta_m^{T}]B_i[\theta_i^{T} \theta_m^{T}]^T\] must be positive where $\theta_i$ is the voltage angle vector of all interior buses in the $i$-th sub-graph, $\theta_m$ is the voltage angle vector of all bus in the vertex cut set and $B_i$ is the reduced susceptance matrix of the $i$-th sub-graph. This means some line power flows in the $i$-th sub-graph must be non-zero, indicating a contradiction to the assumption that the  effects of power injections/withdrawals at bus $j1$ and at the buses in the vertex cut set on the $i$-th sub-graph   precisely offset each other.\end{proof}

\begin{remark}
Corollary  \ref{corr_lmp} points out  an useful property of  vertex cut set  in the trading of power or financial transmission rights (FTRs). Essentially,  the main task in power trading is to find those pairs of buses between which the LMP differences are significant. If a vertex cut set can decompose the  grid such that all congestions are isolated to one or several  sub-grids, then the range of the LMPs of buses in the vertex cut that form  the boundary  of the congested sub-grid(s) will dominate the LMP range of all buses in the non-congested sub-grids. Loosely speaking, the closer the vertex cut gets to the congestions, the higher trading values of the buses in the vertex cut set.  Or in other words, if the congested sub-grid(s)  defined by one given vertex cut set is a subset of the congested sub-grid(s) defined by another vertex cut set, then the buses in the given vertex cut must be of higher or at least equal trading value than that in the other one, suggesting an arbitrage opportunity if violates.
\end{remark}

\section{Grid Decomposition Heuristic}
This section presents a general algorithm structure for grid decomposition heuristics. The objective of the algorithm is   eliminating the portion(s) of  the grid with relatively small range of LMPs from  the consideration of topology control  with
little computational effort. The algorithm is specified in Fig.
 \ref{heu}. This simple algorithm structure first solves an OPF
and, given the OPF results, if the OPF is feasible, finds an  appropriate vertex cut set to decompose the grid into two parts: one with relatively high range of LMPs and the other with relatively low range of LMPs.

The initial vertex cut set that includes the ending buses of all congested lines forms the boundary of the smallest sub-grid containing all congestions. The algorithm then calculates the LMP absolute deviation for each bus in the initial vertex cut set and  locates   the bus, say bus $i$, with the largest LMP absolute deviation.  It then ``pushes out" the initial vertex cut from   bus $i$ by swapping out  bus $i$ (and therefore   bus $i$ becomes an interior bus of the congested sub-grid) and   bringing bus $i$'s neighboring buses into the vertex cut set. 

A pre-set LMP range threshold is the characterizing element  for the  stopping criteria which determines   the number of iterations
applied, and can include a number of conditions. The algorithm terminates if the range of the LMP  of the vertex cut set falls below the pre-set threshold or if all buses  have been brought into the congested sub-grid.   Additionally, the
algorithm may have a pre-set maximum number of iterations.

 \begin{figure}[htbp]
	\centering
		\includegraphics[width=0.3\textwidth]{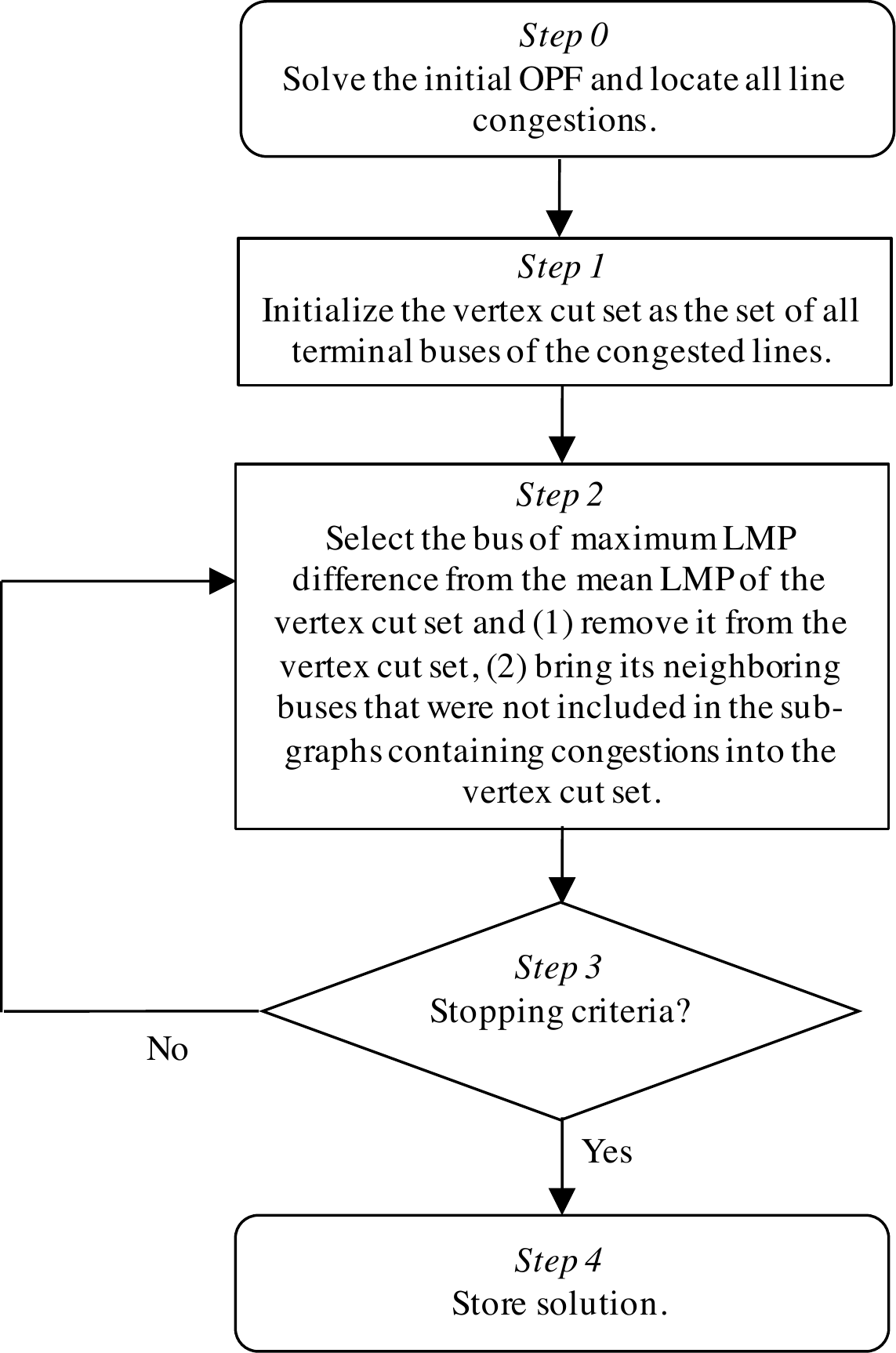}
	\caption{Flow chart describing general algorithm structure of grid decomposition based on vertex cut set.} \label{heu}
\end{figure}
\section{Simulation Results}
To test the validity of the grid decomposition based on vertex cut set, a new topology control  heuristic was proposed and tested on the IEEE 118-bus test systems.  The so called Local Greedy Heuristic  is the same as the Standard Greedy Heuristic in \cite{preprint} except that the switchable lines are limited to the neighboring area of the congested line, i.e the congested sub-grid obtained by the grid decomposition algorithm shown in Fig. \ref{heu}.

\subsection{Test Network Overview}
The heuristics were tested on the IEEE 118-bus test system. This test system represents a portion of the Midwestern US  Power System  as of December, 1962. The data of the test system employed is downloaded from the website of UW$^1$. The generator cost information used in the IEEE 118 network is extracted from  \cite{Blumsack}. The test system consists of 118 buses, 54 generators, and 194 lines, all of
which are assumed to be connected in the initial topology.

\subsection{Standard Greedy Heuristic Overview}
The algorithm is specified in Fig. \ref{standard}. The Standard Greedy Heuristic redoes  the DCOPFs for each  allowable line outage and then select the single line  outage that leads to maximum savings in each iteration.  Please refer \cite{preprint} for more details.

\begin{figure}[htbp]
	\centering
		\includegraphics[width=0.5 \textwidth]{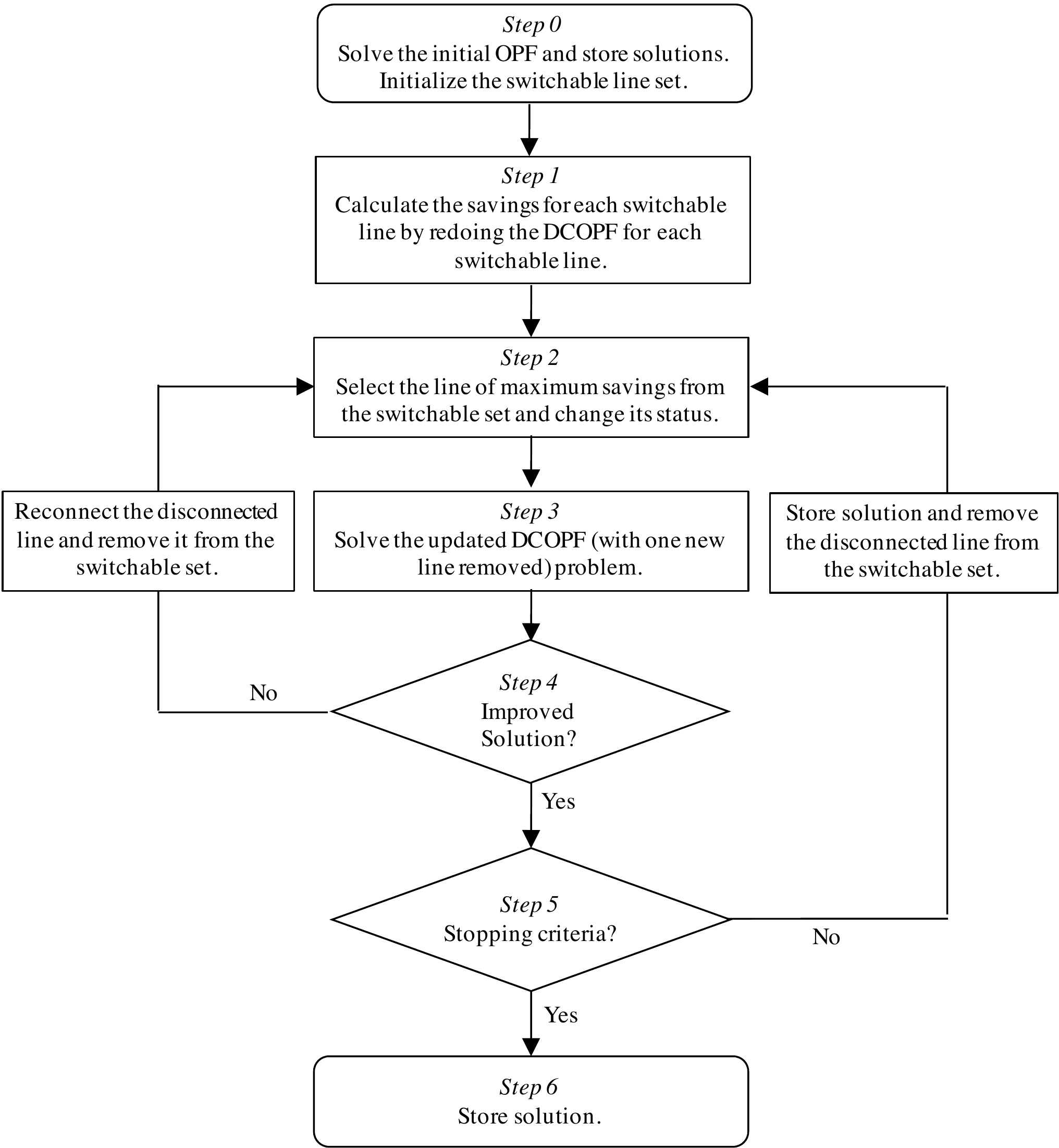}
	\caption{Flow chart describing general algorithm structure of the Standard Greedy Heuristic.} \label{standard}
\end{figure}

Clearly, optimality is not guaranteed, since the transmission topology  is not co-optimized with power dispatch. 
 
 \let\thefootnote\relax\footnotetext{\noindent\underbar{\hspace{0.8in}}\\
$^1$https:$//$www2.ee.washington.edu/research/pstca/pf118/pg$\_$tca118bus.htm
}
\subsection{Simulation Results and Analysis} \label{simlu}
To generate the test scenarios for the heuristics, a fixed load is maintained and we perform a Monte Carlo simulation where the generator costs are randomly varied.   
The sample size used for the Monte Carlo simulation is 50. 
For the Local Greedy Heuristic based on vertex cut set, the LMP range threshold   is set to be 10\% of the maximum LMP difference across the grid. 
%

The percent savings rather than the dollar value are the focus of the paper.
Two benchmark cases are used to evaluate the savings of the heuristic: 1) the one with initial topology and line constraints whose DCOPF sets an upper bound of the total generation costs, and 2) the unconstrained case whose DCOPF provides a lower bound of the generation costs. The cost difference between the two benchmark cases is called the \emph{maximum attainable savings (MAS)}.

The two heuristics are implemented in Matlab, using  MATPOWER \cite{matpower} as the DCOPF solver. The performance comparison between the two heuristic is detailed in Table 2.

\begin{center}\label{compare}
\begin{tabular}{l@{\hskip 0.1in}c@{\hskip 0.05in}c@{\hskip 0.05in}c}\hline
\textbf{Heuristic} &   \textbf{Saving/MAS}  &  \textbf{Lines removed} & \textbf{Mean effort}  \\ \hline \hline 
 
  Local    &  0.653 $\pm$ 0.055  &   12.910 $\pm$ 4.312  & 0.216\\  \hline

Standard  & 0.672 $\pm$ 0.051     &   10.18 $\pm$ 3.98 & 1 \\  \hline

\end{tabular}
\vspace{-0.15in}
\begin{center}
\end{center}
{Table 2: Performance comparison.}
\end{center}

In the worst case, the two heuristics are   of
the same level computational complexity if the grid decomposition algorithm brings all buses into the congested sub-grid. In general, the Standard Greedy Heuristic costs
much more computational effort  especially for a large scale
networks. Though the computational costs of the two heuristics differ  by almost five times, their  attainable savings are quite close, showing the promise in applying grid decomposition to topology control.

\section{Conclusion}
  
In spite of the seemingly promising simulation  results of the heuristic approaches pursued by the research community and a certain level of proven predictability \cite{Baillieul,Wang,Wang_ar} of topology control, satisfactory grid-scale solutions are not guaranteed.  While optimality for this NP-hard problem is no obtained, our approach based on a general spatial grid decomposition nevertheless offers some hope that   it is possible to design algorithms that isolate congestion effects to a relatively small part of the network. By only focusing attention to the neighboring area of the congestions, the search space  of possible switching operations and the computational cost is significantly reduced, and   the locally  optimal switchings at each iteration are still retained to a large degree.  
\bibliography{references}
\bibliographystyle{IEEEtran}
 
\end{document}